\pdfoutput=1
\documentclass[runningheads]{llncs}

\usepackage[utf8]{inputenc}
\usepackage{enumerate}
\usepackage{amsmath,amssymb}
\usepackage{algorithm2e}
\usepackage{chemformula}
\usepackage{lmodern}
\usepackage{url}

\newcommand{\CC}{\mathbb{C}}

\newcommand{\KK}{\mathbb{K}}

\newcommand{\ZZ}{\mathbb{Z}}
\newcommand{\QQ}{\mathbb{Q}}
\newcommand{\RR}{\mathbb{R}}

\begin{document}
\title{A Linear Algebra Approach for Detecting Binomiality of Steady
State Ideals of Reversible Chemical Reaction Networks}

\author{ Hamid Rahkooy\inst{1} \and
Ovidiu Radulescu\inst{2} \and
  Thomas Sturm\inst{1,3}}
\authorrunning{H. Rahkooy, O. Radulescu \& T. Sturm}
\institute{CNRS, Inria, and the University of Lorraine, Nancy, France \\
\email{hamid.rahkooy@inria.fr} \and
LPHI CNRS UMR 5235, and University of Montpellier, Montpellier, France \\
\email{ovidiu.radulescu@umontpellier.fr} \and
MPI-INF and Saarland University, Saarbrücken, Germany\\
\email{thomas.sturm@loria.fr}}
\maketitle

\begin{abstract}
  Motivated by problems from Chemical Reaction Network Theory, we
  investigate whether steady state ideals of reversible reaction
  networks are generated by binomials.  We take an algebraic approach
  considering, besides concentrations of species, also rate constants
  as indeterminates.  This leads us to the concept of unconditional
  binomiality, meaning binomiality for all values of the rate
  constants. This concept is different from conditional binomiality
  that applies when rate constant values or relations among rate
  constants are given.  We start by representing the generators of a
  steady state ideal as sums of binomials, which yields a
  corresponding coefficient matrix. On these grounds we propose an
  efficient algorithm for detecting unconditional binomiality. That
  algorithm uses exclusively elementary column and row operations on
  the coefficient matrix. We prove asymptotic worst case upper bounds
  on the time complexity of our algorithm.  Furthermore, we
  experimentally compare its performance with other existing methods.
\end{abstract}

\section{Introduction}\label{sec:intro}

A \textit{chemical reaction} is a transformation between two sets of
chemical objects called chemical \textit{complexes}. The objects that
form a chemical complex are chemical \textit{species}. In other words,
complexes are formal sums of chemical species representing the left
hand and the right hand sides of chemical reactions.  A
\textit{chemical reaction network} is a set of chemical reactions. For
example
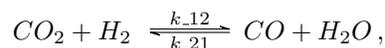
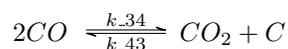
\begin{center}
  \ch[label-style=\scriptsize]{$CO_2+H_2$ <=>[$k_{12}$][$k_{21}$]
    $CO+H_2O$}, 
\end{center}
\begin{center}
    \ch[label-style=\scriptsize]{$2CO$ <=>[$k_{34}$][$k_{43}$] $CO_2+C$}
\end{center}
is a chemical reaction network with two \textit{reversible} reactions.

A \textit{kinetics} of a chemical reaction network is an assignment of
a rate function, depending on the concentrations of chemical species
at the left hand side, to each reaction in the network.  A kinetics
for a chemical reaction network is called \textit{mass-action} if for
each reaction in the chemical reaction network, the rate function is a
monomial in the concentrations of the chemical species with exponents
given by the numbers of molecules of the species consumed in the
reaction, multiplied by a constant called rate constant.  Reactions
are classified as zero-order, first-order, etc. according to the order
of the monomial giving the rate.  For reversible reactions, the net
reaction rate is a binomial, the difference between the forward and
backward rates.  In the example above $k_{12}$, $k_{21}$, $k_{23}$,
$k_{32}$ are the rate constants. In this article we generally assume
mass-action kinetics. We furthermore assume that reactions are
reversible, unless explicitly specified otherwise.

The change in the concentration of each species over time in a
reaction can be described via a system of autonomous ordinary
differential equations. For instance, consider the chemical reaction
network above and let $x_1$, $x_2$, $x_3$, $x_4$, $x_5$ be the
indeterminates representing the concentrations of the species $CO_2$,
$H_2$, $CO$, $H_2O$ and $C$, respectively.  The corresponding
differential equations are
\begin{align}
  \dot{x}_1  &=p_1, & p_1 & =-k_{12}x_1x_2+ k_{21}x_3x_4 -
                            k_{34}x_3^2-k_{43}x_1x_5, \label{ode21} \\
  \dot{x}_2  &=p_2, & p_2& = -k_{12}x_1x_2+ k_{21}x_3x_4, \label{ode22} \\    
  \dot{x}_3 &=p_3, & p_3&=  k_{12}x_1x_2- k_{21}x_3x_4 +
                          -2k_{34}x_3^2+2k_{43}x_1x_5, \label{ode23} \\
  \dot{x}_4  &=p_4, & p_4& = k_{12}x_1x_2- k_{21}x_3x_4 , \label{ode24} \\    
  \dot{x}_5  &=p_5, & p_5& = -k_{12}x_1x_2+ k_{21}x_3x_4 +
                           k_{34}x_3^2-k_{43}x_1x_5. \label{ode25} 
\end{align}
Each zero of the polynomials $p_1$, $p_2$, $p_3$, $p_4$, $p_5$ gives a
concentration of species in which the system is in equilibrium. The
zeros of $p_1$, $p_2$, $p_3$, $p_4$, $p_5$ are called the steady
states of the chemical reaction network. Accordingly, the ideal
$\langle p_1,p_2.p_3,p_4,p_5 \rangle \subseteq
\QQ[k_{12},k_{21},k_{34},k_{43},x_1,x_2,x_3,x_4,x_5]$ is called the
{\em steady state ideal} of the chemical reaction network. We consider
the coefficient field $\QQ$ because of computability issue. Otherwise,
theoretically, our results hold for any coefficient field. The
solutions of these polynomials can be in $\RR$ or in $\CC$.

For a thorough introduction to chemical reaction network theory, we
refer to Feinberg's Book \cite{Feinberg-Book} and his lecture notes
\cite{Feinberg-Lectures}. We follow the notation of Feinberg's book in
this article.
  
An ideal is called binomial if it is generated by a set of
binomials. In this article we investigate whether the steady state
ideal of a given chemical reaction network is binomial. We are
interested in efficient algorithms for testing binomiality. Consider
the steady state ideal
\begin{equation}
I = \langle p_1,p_2,p_3,p_4,p_5\rangle \subseteq
\QQ[k_{12},k_{21},k_{34},k_{43},x_1,x_2,x_3,x_4,x_5],
\end{equation}
given by Equations (\ref{ode21})--(\ref{ode25}).  Reducing $p_1$,
$p_3$ and $p_4$ with respect to $p_2$ and $p_5$, we have
\begin{equation}\label{eq:steadystateideal}
  I = \langle -k_{12}x_1x_2+ k_{21}x_3x_4, -k_{34}x_3^2+k_{43}x_1x_5 \rangle, 
\end{equation}
which shows that the ideal $I$ is binomial. In this article, we work
over the ring $\QQ[k_{ij},x_1,\dots,x_n]$ and investigate binomiality
over this ring.

Note that in the literature there exist also slightly different
notions of binomiality. Eisenbud and Sturmfels in
\cite{eisenbud1996binomial} call an ideal binomial if it is generated
by polynomials with at most two terms. Following this definition, some
authors, e.g., Dickenstein et al.~in \cite{perez_millan_chemical_2012}
have considered the steady state ideal as an ideal in the ring
$\QQ(k_{ij})[x_1,\dots,x_n]$ and studied the binomiality of these
ideals in $\RR[x_1,\dots,x_n]$ after specialising $k_{ij}$ with
positive real values.  In order to distinguish between the two
notions, we call {\em unconditionally binomial} a steady state ideal
that is binomial in $\QQ[k,x]$ (the notion used in this paper) and
{\em conditionally binomial} a steady state ideal that is binomial in
$\QQ(k)[x]$, i.e.  for specified parameters $k$ (the notion used in
\cite{perez_millan_chemical_2012}).

The notions of binomial ideals and toric varieties have roots in
thermodynamics, dating back to Boltzmann. Binomiality corresponds to
detailed balance, which for reaction networks means that at
thermodynamic equilibrium the forward and backward rates should be
equal for all reactions. Detailed balance is a very important concept
in thermodynamics, for instance it has been used by Einstein in his
Nobel prize winning theory of the photoelectric effect
\cite{einstein1916strahlungs}, by Wegscheider in his thermodynamic
theory of chemical reaction networks \cite{Wegscheider1901} and by
Onsager for deriving his famous reciprocity relations
\cite{onsager1931reciprocal}.  Because detailed balance implies time
reversal symmetry, systems with detailed balance can not produce
directed movement and can only dissipate heat. This is important in
applications, for instance in molecular biology, where molecular
motors can not function with detailed balance. Although most
interesting molecular devices function without detailed balance and
binomiality, some of their subsystems can satisfy these
conditions. The interest of studying binomiality relies in the
simplicity of the analysis of such subsystems. For instance, important
properties such as multistationarity and stability are easier to
establish for binomial systems.  Toricity, also known as complex, or
cyclic, or semi-detailed balance is also known since Boltzmann that
has used it as a sufficient condition for deriving his famous
H-theorem \cite{boltzmann1964lectures}. Binomiality implies toricity,
but the converse is not true: in order to have binomiality, a toric
system must obey constraints on the rates constants, such as the well
known Weigscheider-Kolmogorov condition asking for the equality of the
products of forward and backward rates constants in cycles of
reversible reactions. In this paper we focus on the situation when
detailed balance is satisfied without conditions on the rate
constants.

Detecting binomiality of an ideal, particularly of a steady state
ideal, is a difficult problem, both from a theoretical and a practical
point of view. The problem is typically solved by computing a
Gr\"obner basis, which is EXPSPACE-complete
\cite{mayr_complexity_1982}. Recent linear algebra approaches for
solving the problem in a different setting than our problem construct
large matrices which also points at the difficulty of the problem
\cite{millan2012chemical,conradi2015detecting}.

There is quite comprehensive literature on chemical reaction network
theory. An excellent reference to this topic is
\cite{Feinberg-Book,Feinberg-Lectures}. As mathematical concepts,
binomiality and toricity have been widely studied and their properties
have been investigated by various authors, e.g., Fulton
\cite{fulton_introduction_2016}, Sturmfels
\cite{sturmfels_grobner_1996}, Eisenbud et
al.~\cite{eisenbud1996binomial}. Binomiality and toricity show up
quite often in chemical reaction networks. Binomiality in the case of
\textit{detailed balancing} of reversible chemical reactions has been
studied by Gorban et
al.~\cite{gorban_generalized_2015,gorban_three_2015} and Grigoriev and
Weber \cite{GrigorievWeber2012a}.  Feinberg \cite{Feinberg1972} and
Horn and Jackson \cite{Horn1972} have studied toric dynamical systems.
Gatermann et al.~studied \textit{deformed toricity} in
\cite{Gatermann2000}. Craciun, et al.~have considered the toricity
problem over the real numbers in \cite{craciun_toric_2009} and have
presented several interesting results in this regard, among them, they
have shown that \textit{complex balanced systems} are the same as
toric dynamical systems, although \textit{toric steady states} are
different from that.  It has been shown in
\cite{dickenstein2019,sadeghimanesh2019} that the binomial structure
will imply much simpler criteria for multistationarity. These results
give strong motivation for one to study algorithms for detecting
binomial networks. Especially, in \cite{dickenstein2019}, the authors
defined \textit{linearly binomial network} and they proposed
sufficient conditions for a network to be linearly binomial. The proof
is constructive even though it has not been presented as an
algorithm. Their method is also quite straightforward and can handle
more general networks in many applications.

Dickenstein et al.~have presented sufficient linear algebra conditions
with inequalities for binomiality of the steady state ideals in
\cite{millan2012chemical}. Their idea has been developed in
\cite{millan_structure_2018}, where the concept of MESSI reactions has
been introduced. Conradi and Kahle have proved in
\cite{conradi2015detecting} that for homogenous ideals (i.e. for
chemical reaction networks without zero-order reactions), the
sufficient condition of Dickenstein et al.~is necessary as well and
also introduced an algorithm for testing binomiality of homogenous
ideals. As many biochemical networks are not homogeneous, the
algorithm requires heuristics in such cases. The algorithm has been
implemented in Maple and Macaulay II in
\cite{MapleCK,alex2019analysis} and experiments have been carried out
on several biological models. Grigoriev et al.~in
\cite{grigoriev2019efficiently} have considered the toricity of steady
state ideals from a geometric point of view. Introducing shifted
toricity, they presented algorithms, complexity bounds as well as
experimental results for testing toricity using two important tools
from symbolic computation, quantifier elimination
\cite{Davenport:1988:RQE:53372.53374,Grigoriev:88a,Weispfenning:88a}
and Gr\"obner bases
\cite{Buchberger:65a,bb-system,Faugere:99a,Faugere:02a}.  Recently,
first order logic test for toricity have been introduced
\cite{sturm2020firstorder}.

The main idea of this article is to consider the generators of the
steady state ideal as sums of the binomials associated to the
reactions rather than the monomials associated to the complexes. This
is feasible for a reversible chemical reaction network. Following the
above observation and assigning a binomial to each reaction, one can
write the generators of the steady state ideal as sums of those
binomials with integer coefficients.

As our main result, we have proved that a reversible chemical reaction
network is unconditionally binomial if and only if it is “linearly”
binomial (i.e., there exist linear combinations of the generators such
that these combinations are binomials).  More precisely, having
represented of the generators of the steady state ideal as sum of
binomials, one can test the binomiality exclusively using elementary
row and column operations on the coefficient matrix of these
binomials. This can be done by computing the reduced row echelon form
of the coefficient matrix, which yields an efficient method for
testing binomiality.

Our main contributions in this article are the following.
\begin{enumerate}
\item We introduce a new representation of the generators of the steady state
  ideal of a reversible chemical reaction as a sum of certain binomials rather
  than monomials.
\item Using that representation, we assign a matrix with entries in $\ZZ$ to a
  reversible chemical reaction network, such that the binomiality of the steady
  state ideal can be tested by computing the reduced row echelon form of this matrix.
\item We prove a worst-case upper bound on the time complexity of our
  binomiality test. We experimentally compare our test with the existing
  binomiality tests in the literature, which demonstrates the applicability of
  our method.
\end{enumerate}

Our representation of the steady state ideal as a sum of certain
binomials, as well as the matrices associated to them are further
original ideas presented in this paper. While typically
complex-species matrices are used for testing binomiality, we use
reaction-species matrices for this purpose.

The plan of the article is as follows. Section \ref{sec:intro} gives
an introduction to the necessary concepts of chemical reaction network
theory, reviews the literature and presents the idea of this
work. Section \ref{sec:linearalgebra-algo} includes the main
definitions and results. In this section we show our representation of
the generators of the steady state ideal of a reversible chemical
reaction network and present our algorithm for testing binomiality. In
Section \ref{sec:complexity}, we discuss the complexity of our
method. We furthermore compare our algorithm with other existing
algorithms in the literature via experiments. In Section
\ref{sec:conclusion} we summarise our results and draw some
conclusions.

\section{Testing Binomiality}\label{sec:linearalgebra-algo}
In this section, we present our main result based on which we present
an algorithm for testing unconditional binomiality of reversible
chemical reaction networks. In Subsection \ref{subsec:sum-binom-rep}
we introduce a representation for the generators of the steady state
ideal of a chemical reaction network as sum of binomials. We show that
this representation is unique for reversible reaction networks,
considering rate constants as indeterminates. In Subsection
\ref{subsec:algorithm}, we define a matrix associated to a chemical
reaction network which is essentially the species--reaction matrix,
rather than the stoichiometric matrix which is the species--complex
matrix. Having considered constant rates as indeterminates, the
uniqueness of our matrix for reversible reactions comes from the
uniqueness of representing of the generators of the steady state ideal
as sum of binomials.

\subsection{Sum of Binomial
  Representation}\label{subsec:sum-binom-rep}
Consider the following reversible reaction between two complexes $C_1$
and $C_2$.
\begin{center}
  \ch[label-style=\scriptsize]{C_1 <=>[$k_{12}$][$k_{21}$] C_2}.
\end{center}
Let $m_i$, $i=1,2$, be the product of the concentrations of the
species in $C_i$ with the stoichiometric coefficients as the powers.
We call $m_i$ the monomial associated to $C_i$. Also let $x_1$ be the
concentration of a species that is in $C_1$ with the stoichiometric
coefficient $\alpha_1$ and is not in $C_2$. The differential equation
describing the kinetics of this species is
\begin{equation}\label{rec-bin1}
  \dot{x}_1=-\alpha_1k_{12}m_1+ \alpha_1k_{21}m_2. 
\end{equation}
For a species in $C_2$ with stoichiometric coefficient $\alpha_2$
which is not in $C_1$ with the concentration $x_2$, the differential
equation will be
\begin{equation}\label{rec-bin2}
  \dot{x}_2=\alpha_2 k_{12}m_1 - \alpha_2 k_{21}m_2.
\end{equation}
For a species with concentration $x_3$ that appears in both $C_1$ and
$C_2$, the differential equation will be
$\dot{x}_3=c(k_{12}m_1- k_{21}m_2)$, where $c \in \ZZ$ is the
difference between the corresponding stoichiometric coefficients in
$C_2$ and $C_1$. Set $b_{12}:=-k_{12}m_1+k_{21}m_2$ and
$b_{21}:=k_{12}m_1-k_{21}m_2$.  The steady state ideal of the above
chemical reaction network is
$\langle \alpha_1b_{12}, \alpha_2b_{21}\rangle$, which is equal to
$\langle b_{12}\rangle$, since $b_{12}=-b_{21}$.

For a reversible reaction network with more than one reaction, one can
associate a binomial of the form $b_{ij}:=k_{ij}m_i-k_{ji}m_j$ to each
reaction. Then the polynomials generating the steady state ideal can
be written as sums of $b_{ij}$ with integer coefficients. We make this
more precise in the following definition.
\begin{definition}\label{def:binom-rep}
  Let $\mathcal{C}$ be a reversible chemical reaction network with the
  complexes $C_1,\dots,C_s$, let $k_{ij}$, $1 \leq i \ne j \leq s$, be
  the rate constant of the reaction from $C_i$ to $C_j$, and let
  $x_1,\dots,x_n$ be the concentrations of the species in the chemical
  reaction network. We call a monomial $m_i$ the \textit{monomial
    associated to} $C_i$ if $m_i$ is the product of the concentrations
  of those species that appear in $C_i$ with the stoichiometric
  coefficients of the species as the powers. If there is a reaction
  between $C_i$ and $C_j$, then $b_{ij}:=-k_{ij}m_i+k_{ji}m_j$ is
  called the \textit{binomial associated to the reaction from $C_i$ to
    $C_j$}, otherwise $b_{ij}:=0$.
\end{definition}

\begin{example}\label{ex:binom-sum}
  Recall the following chemical reaction network form Section
  \ref{sec:intro}:
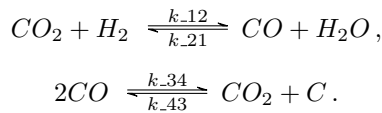
  \begin{center}
    \ch[label-style=\scriptsize]{$CO_2+H_2$ <=>[$k_{12}$][$k_{21}$]
      $CO+H_2O$},
  \end{center}
  \begin{center}
    \ch[label-style=\scriptsize]{$2CO$ <=>[$k_{34}$][$k_{43}$]
      $CO_2+C$}.
  \end{center}
  Following the notation in Section \ref{sec:intro}, let
  $x_1,x_2,x_3,x_4,x_5$ be the concentrations of $CO_2$, $H_2$, $CO$,
  $H_2O$ and $C$, respectively. The monomials associated to the
  complexes $CO_2+H_2$, $CO+H_2O$, $2CO$ and $CO_2+C$ are $x_1x_2$,
  $x_3x_4$, $x_3^2$ and $x_1x_5$, respectively. The binomials
  associated to the two reactions in this network are
  $b_{12}=-k_{12}x_1x_2+k_{21}x_3x_4$ and
  $b_{34}=-k_{34}x_3^2+k_{43}x_1x_5$. As there is no reaction between
  the first and third complexes we have $b_{13}=b_{31}=0$. Similarly,
  $b_{23}=b_{32}=0$, $b_{14}=b_{41}=0$ and $b_{24}=b_{42}=0$. Also, by
  definition, $b_{21}=-b_{12}$, $b_{34}=-b_{43}$, etc.. Using the
  binomials associated to the reactions, one can write the polynomials
  generating the steady state ideal as
\begin{align}
  & p_1 =  b_{12} -b_{34} \label{pol21} \\
  & p_2 =  b_{12}  \label{pol22} \\    
  & p_3 =  -b_{12}+2b_{34} \label{pol23} \\
  & p_2 =  -b_{12}  \label{pol24} \\    
  & p_2 =  -b_{34}  \label{pol25}.   
\end{align}
Hence, the steady state ideal can be written as
\begin{equation}\label{eq:ssi1-binom-basis}
  \langle p_1,p_2, p_3,p_4,p_5 \rangle  = \langle b_{12}, b_{34} \rangle. 
\end{equation}
\end{example}

As Example \ref{ex:binom-sum} and the definition of the binomials
$b_{ij}$ in Definition \ref{def:binom-rep} suggests one can write the
generators of the steady state ideal of every reversible chemical
reaction networks as sums of $b_{ij}$ with integer coefficients, i.e.,
assuming that $\mathcal{R}$ is the set of reactions in the chemical
reaction network
\begin{equation}\label{eq:ssi-binom-basis}
  \dot{x}_k =p_k=\sum\limits_{C_i \rightarrow C_i \in \mathcal{R}}
  c_{ij}^{(k)} b_{ij}, 
\end{equation}
for $k=1\ldots n$ and $c_{ij}^{(k)} \in \ZZ$. 

For clarification, we may remind the reader that in this article we
assume working over $\QQ[k_{ij},x_1,\dots,x_n]$. This is the case, in
particular, for Definition \ref{def:binom-rep} and the discussion
afterwards. In \cite{perez_millan_chemical_2012}, the authors
specialise $k_{ij}$ with positive real values, in which case, the
steady state ideal may or may not be binomial over
$\RR[x_1,\dots,x_n]$. Similarly, specialising $k_{ij}$ in Equation
\ref{eq:ssi-binom-basis} can result in writing $p_k$ as sum of
different binomials. In other words, if $k_{ij}$ specialised, the
representation of $p_k$ as sum of binomials in
\ref{eq:ssi-binom-basis} is not necessarily unique. This is
illustrated in the following example.

\begin{example}\cite[Example
  2.3]{perez_millan_chemical_2012}\label{ex:dickenstein}
  Let $C_1 = 2A$, $C_2=2B$ and $C_{3}=A+B$. Consider the reversible
  chemical reaction network given by the following reactions:
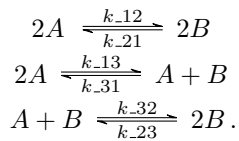
\begin{center}
  \ch[label-style=\scriptsize]{$2A$ <=>[$k_{12}$][$k_{21}$] $2B$}\\[0.5ex]
  \ch[label-style=\scriptsize]{$2A$<=>[$k_{13}$][$k_{31}$] $A+B$}\\[0.5ex]
  \ch[label-style=\scriptsize]{$A+B$<=>[$k_{32}$][$k_{23}$] $2B$}.
\end{center}
Assuming $x_1$ and $x_2$ to be the concentrations of $A$ and $B$,
respectively, by Definition~\ref{def:binom-rep},
\begin{align}
  b_{12} &= -k_{12}x_1^2 + k_{21}x_2^2 \\
  b_{13} &=  -k_{13}x_1^2 +k_{31}x_1x_2 \\
  b_{23} &=  k_{23}x_2^2 -k_{32}x_1x_2.
\end{align}
It can be checked that the generators of the steady state ideal can be
written as
\begin{align}
 p_1 &=  2b_{12}+b_{13}+b_{23} \\
 p_2 &=  -2b_{12}-b_{13}-b_{23}.
\end{align}
If $k_{31}=k_{32}$ then $k_{31}x_1x_2=k_{32}x_1x_2$, hence
$k_{31}x_1x_2$ will occur in $b_{13}$ and $b_{23}$ with opposite signs
which will be cancelled out in $b_{13}+b_{23}$, resulting in writing
$p_1$ as sum of $b_{12}$ and $-k_{13}x_1^2+k_{23}x_2^2$. This is
another way of writing $p_1$ as sum of binomials. Because binomiality
relies here on the condition $k_{31}=k_{32}$, this is an example of
conditional binomiality.
\end{example}

If we consider the rate constants $k_{ij}$ as indeterminates, i.e., if
we consider the steady state ideal as an ideal over the ring
$\QQ[k_{ij},x_1,\dots,x_n]$, then the representation in Equation
(\ref{eq:ssi-binom-basis}) as sum of binomials $b_{ij}$ will be
unique. More precisely, we have the following.

\begin{lemma}\label{gen-cond1}
  Given a reversible chemical reaction network with the notation of
  Definition \ref{def:binom-rep}, if $k_{ij}$ are indeterminates then
  the generators of the steady state ideal can be uniquely written as
  sum of the binomials presented in Equation \ref{eq:ssi-binom-basis}.
\end{lemma}
\begin{proof}
  Assuming that $k_{ij}$, $1 \leq i,j \leq s$ are indeterminates, they
  will be algebraically independent over
  $\QQ[x_1,\dots,x_n]$. Therefore for monomials $m_t$ and $m_{t'}$ in
  $\QQ[x_1,\dots,x_n]$ associated to two distinct complexes and for
  all $1 \leq i,j,i',j' \leq~s$, $k_{ij}m_t$ and $k_{i'j'}m_{t'}$ will
  be distinct monomials in $\QQ[k_{ij}, x_1,\dots,x_n]$.  Hence
  binomials $b_{ij}$ associated to the reversible reactions are not
  only pairwise distinct, but also their monomials are pairwise
  distinct in $\QQ[k_{ij},x_1,\dots,x_n]$. This implies that the
  generators of the steady state ideal have unique representations in
  $\QQ[k_{ij},x_1,\dots,x_n]$ as sum of $b_{ij}$ with integer
  coefficients.
\end{proof}

Having a unique representation as in Equation \ref{eq:ssi-binom-basis}
enables us to represent our binomial coefficient matrix, defined
later, which is the base of our efficient algorithm for testing
unconditional binomiality of reversible chemical reaction networks.

Considering rate constants $k_{ij}$ as indeterminates, if a steady
state ideal is unconditionally binomial, i.e., binomial in the ring
$\QQ[k_{ij},x_1,\dots,x_n]$, then its elimination ideal is binomial in
the ring $\QQ[x_1,\dots,x_n]$. Indeed, the elimination of a binomial
ideal is a binomial ideal. This can be seen from elimination property
of Gr\"obner bases. Authors of \cite{eisenbud1996binomial} have
studied binomial ideals and their properties intensively. In
particular Corollary 1.3 in the latter article state the binomiality
of the elimination ideal of a binomial ideal. We remind the reader
that the definition of binomiality in this article is different from
\cite{eisenbud1996binomial}. In the latter, binomial ideals have
binomial and monomial generators, however in the current article, we
only consider binomial generators. Restricting the definition of
binomial ideal to the ideals with only binomial generators, most of
the result in \cite{eisenbud1996binomial} still holds, in particular
the one about the elimination of binomial ideals. Therefore, if the
steady state ideal of a chemical reaction network is binomial in
$\QQ[k_{ij},x_1,\dots,x_n]$, then its elimination
$I \cap \QQ[x_1,\dots,x_n]$ is also a binomial ideal.

Geometrically, the above discussion can be explained via projection of
the corresponding varieties. Given a chemical reaction network, assume
that reaction rates $k_{ij}$ are indeterminates and let the number of
$k_{ij}$ be $t$. Let $V$ denote the steady state variety, i.e., the
variety of the steady state ideal.  $V$ is a Zariski closed subset of
${\KK}^{t+n}$, where $\KK$ is an appropriate field (e.g., $\CC$). If
$V$ is a coset of a subgroup of the multiplicative group
$(\KK^*)^{t+n}$, then the projection of $V$ into the space generated
by $x_1\dots,x_n$, i.e., $V \cap (\KK^*)^n$ is also a coset. In
particular, the projection of a group is a group. Since the variety of
a binomial ideal is a coset
\cite{grigoriev2019efficiently,grigoriev_milman2012}, the projection
of the variety of a binomial ideal is the variety of a binomial
ideal. As special cases, the projection of a toric variety, a shifted
toric variety and a binomial variety (defined in
\cite{grigoriev2019efficiently,grigoriev_milman2012}) is a toric, a
shifted toric and a binomial variety, respectively. For a detailed
study of toricity of steady state varieties, we refer to
\cite{grigoriev2019efficiently}.

\begin{remark}
  \begin{itemize}
  \item We may mention that in \cite{craciun_toric_2009}, the authors
    have studied \textit{toric dynamical systems}, where they have
    considered working over $\QQ[k_{ij},x_1,\dots,x_n]$ and presented
    several interesting results. In particular, Theorem 7 in that
    article states that a chemical reaction network is toric if and
    only if the rate constants lie in the variety of a certain ideal
    in $\QQ[k_{ij}]$, called the \textit{moduli ideal}.
  \item Toric dynamical systems are known as \textit{complex
      balancing} mass action systems \cite{craciun_toric_2009}.
  \end{itemize}
\end{remark}

\subsection{The Algorithm}\label{subsec:algorithm}
\begin{definition}\label{def:binom-matrix}
  Let $\mathcal{C}$ be a reversible chemical reaction network as in
  Definition \ref{def:binom-rep} and assume that the generators of its
  steady state ideal are written as the linear combination of the
  binomials associated to its reactions as in Equation
  \ref{eq:ssi-binom-basis}, i.e.,
  \begin{displaymath}
    p_k=\sum_{C_i \rightarrow C_i \in \mathcal{R}}^{s} c_{ij}^{(k)} b_{ij}\quad
    \text{for}\quad  k=1,\dots,n.
  \end{displaymath}
  We define the \textit{binomial coefficient matrix} of $\mathcal{C}$
  to be the matrix whose rows are labeled by $p_1,\dots,p_n$ and whose
  columns are labeled by non-zero $b_{ij}$ and the entry in row $p_k$
  and column $b_{ij}$ is $c_{ij}^{(k)} \in \ZZ$.
\end{definition}

By the definition, the binomial coefficient matrix of a reversible
chemical reaction network is the coefficient matrix of the binomials
that occur in the representation of the generators of the steady state
ideal as sum of binomials.  As we consider $k_{ij}$ indeterminates,
the representation of the generators of the steady state ideal of a
given complex is unique, which implies that the binomial coefficient
matrix of a given complex is unique too.

\begin{example}\label{ex:binom-matrix}
  Consider the chemical reaction network in Example
  \ref{ex:binom-sum}, with generators of the steady state ideal as
  follows.
  \begin{align}
  & p_1 =  b_{12} -b_{34} \\
  & p_2 =  b_{12}  \\    
  & p_3 =  -b_{12}+2b_{34} \\
  & p_2 =  -b_{12}  \\    
  & p_2 =  -b_{34}.   
  \end{align}
  The binomial coefficient matrix of this chemical reaction network is
  \begin{equation}\label{eq:binom-matrix}
    M \ = \ \ \, \bordermatrix{ & b_{12} & b_{34} \cr
      p_1 & 1 & -1 \cr
      p_2 & 1 & 0 \cr
      p_3 & -1 & 2 \cr
      p_4 & -1 & 0 \cr
      p_5 & 0 & -1 \cr}.
  \end{equation}
\end{example}

Another simple example is the reaction
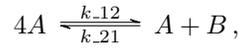
\begin{center}
  \ch[label-style=\scriptsize]{$4A$<=>[$k_{12}$][$k_{21}$] $A+B$},
\end{center}
with the binomial associated to it as
$b_{12}:=-k_{12}x_1^4+k_{21}x_1x_2$, where $x_1$ is the concentration
of $A$ and $x_2$ is the concentration of $B$. The steady state ideal
is generated by $\{ 3b,-b \}$, and the binomial coefficient matrix for
this network is $\genfrac{(}{)}{0pt}{}{3}{-1}$.

One can test binomiality of the steady state ideal of a reversible
reaction network using its binomial coefficient matrix.

\begin{theorem}\label{thm:hermite-form} 
  The steady state ideal of a reversible chemical reaction network is
  unconditionally binomial, i.e., binomial in
  $\QQ[k_{ij},x_1,\dots,x_n]$, if and only if the reduced row echelon
  form of its binomial coefficient matrix has at most one non-zero
  entry at each row.
\end{theorem}
\begin{proof}
  Let $G=\{p_1,\dots,p_n\} \subseteq \QQ[k_{ij},x_1,\dots,x_n]$ be a
  generating set for the steady state ideal of a given reversible
  chemical reaction network $\mathcal{C}$, and let
  $\{b_{ij} \mid 1 \leq i \ne j \leq s\}$ be the ordered set of
  non-zero binomials associated to the reactions. Fix a term order on
  the monomials in $\QQ[k_{ij},x_1,\dots,x_n]$.

  First we prove that if the reduced row echelon form of the binomial
  coefficient matrix has at most one non-zero entry at each row, then
  the steady state ideal is binomial. The proof of this side of the
  proposition comes from the definition of reduced row echelon
  form. In fact, the reduced row echelon form of the binomial
  coefficient matrix of $\mathcal{C}$ can be computed by row reduction
  in that matrix, which is equivalent to the reduction of the
  generators of the steady state ideal with respect to each
  other. Therefore, computing the reduced row echelon form of the
  binomial coefficient matrix and multiplying it with the vector of
  binomials $b_{ij}$, one can obtain another basis for the steady
  state ideal. Having this, if the reduced row echelon form has at
  most one non-zero entry at each row, then the new basis for the
  steady state ideal will only include $b_{ij}$. Therefore the steady
  state ideal will be binomial.

  Now we prove the ``only if'' part of the proposition, that is, if
  the steady state ideal of $\mathcal{C}$ is binomial, then the
  reduced row echelon form of the binomial coefficient matrix has at
  most one non-zero entry at each row. We claim that for each pair of
  polynomials $p_t, p_m \in G$, $p_t$ is reducible with respect to
  $p_m$ if and only if there exists a binomial $b_{ij}$ that occurs in
  both $p_t$ and $p_m$ and includes their leading terms. The ``only
  if'' part of the claim is obvious. To prove the ``if'' part of the
  claim, let $p_m$ be reducible with respect to $p_t$. Then the
  leading term of $p_m$ divides the leading term of $p_t$. Since the
  leading terms are multiples of $k_{ij}$ and these are disjoint
  indeterminates, this is only possible if both of the leading terms
  are equal. If the leading terms are equal, then $b_{ij}$ in which
  the leading terms occur, must itself occur in both $p_t$ and $p_m$.
  Therefore $p_t$ and $p_m$ share a binomial associated to a reaction,
  which is in contradiction with our assumption.

  From the above claim and the definition of the reduced row echelon
  form one can see that $p_1,\dots,p_n$ are pairwise irreducible if
  and only if the binomial coefficient matrix of $\mathcal{C}$ is in
  reduced row echelon form.

  Now we prove that $p_1,\dots,p_n$ are pairwise irreducible if and
  only if they form a Gr\"obner basis in which polynomials are
  pairwise irreducible. Note that this does not necessarily imply that
  $G$ is a a reduced Gr\"obner basis, as $p_i$ are not necessarily
  monic. Assume that $p_1,\dots,p_n$ are pairwise irreducible. We
  prove that the greatest common divisor of each pair of the leading
  terms of the $p_1,\dots,p_n$ is $1$. By contradiction, assume that
  there exists a monomial not equal to $1$ which divides the leading
  terms of both $p_t$, $p_m$, for $1 \leq t,m \leq n$. Then there
  exists a variable $x_l$ such that $x_l$ divides the leading terms of
  $p_t$ and $p_m$. Since each leading term is the monomial associated
  to a complex, the species with concentration $x_1$ occurs in two
  complexes with associated monomials as the leading terms of $p_t$
  and $p_m$. Then both $p_t$ and $p_m$ have as their summand the
  binomials that are associated to the reactions including those
  complexes. As for each complex there exists at least one binomial
  associated, both $p_m$ and $p_t$ have as a summand one common
  binomial $b_{ij}$. However, we had already proved that this implies
  that $p_t$ and $p_m$ are not pairwise irreducible, which is a
  contradiction to the assumption that the greatest common divisor of
  the leading terms of $p_t$ and $p_m$ is not $1$. Now by Buchberger's
  first criterion if the greatest common divisor of the leading terms
  of each pair of polynomials in $G$ is $1$ then $G$ is a Gr\"obner
  basis.  The other side of this claim is obvious.

  From what we have proved until now, we can conclude that the
  binomial coefficient matrix of $\mathcal{C}$ is in reduced row
  echelon form if and only if $G$ is a Gr\"obner basis with pairwise
  irreducible elements. On the other hand, by a result of Eisenbud and
  Sturmfels \cite{eisenbud1996binomial}, the steady state ideal of
  $\mathcal{C}$ is binomial if and only if every Gr\"obner basis of it
  includes binomials. Therefore we conclude that the steady state
  ideal is binomial if and only if the reduced row echelon form of the
  binomial coefficient matrix has at most one non-zero entry in each
  row.
\end{proof}

\begin{example}
  Following Example \ref{ex:binom-matrix}, one case easily see that
  the reduced row echelon form of the binomial coefficient matrix
  (\ref{eq:binom-matrix}) is
  \begin{equation}
    M \ = \ \ \, \bordermatrix{ & b_{12} & b_{34} \cr
      p_1 & 1 & 0 \cr
      p_2 & 0 & 1 \cr
      p_3 & 0 & 0 \cr
      p_4 & 0 & 0 \cr
      p_5 & 0 & 0 \cr},
  \end{equation}
  which means that the steady state ideal is unconditionally binomial
  and is generated by $\{b_{12},b_{34}\}$.
\end{example}

Theorem \ref{thm:hermite-form} yields Algorithm \ref{alg:binom-test}
for testing unconditional binomiality. The input of the algorithm is a
reversible chemical reaction network, given by the vector of monomials
associated to its complexes, $(m_1,\dots,m_s)$, and the rates
$k_{ij}$. It uses a function \textit{IsBinomial} which takes a set of
polynomials and checks if all of them are binomial.

\begin{algorithm}[t]
  \caption{Testing Unconditional Binomiality of Reversible Chemical
    Reaction Networks \label{alg:binom-test}}
  \DontPrintSemicolon
  \SetAlgoVlined
  \LinesNumbered
  \SetKwProg{Fn}{Function}{}{end}
  \SetKwFunction{BINOMTEST}{BinomialityTest}
  \SetKwFunction{RREF}{ReducedRowEchelonForm}
  \SetKwFunction{MAT}{Matrix}
  \SetKwFunction{SET}{MakeSet}
  \SetKwFunction{IB}{IsBinomial}
  
  \Fn{\BINOMTEST{$\mathcal{C}$}}{
    \KwIn{$\mathcal{C}=\{(m_1,\dots,m_s)\in [X]^n, k_{ij}\}$}
    \KwOut{Binomial or NotBinomial}
    $b_{ij}:=-k_{ij}m_i+k_{ji}m_j, 1\leq i \ne j \leq s$\;
    $B:=(b_{ij}, 1\leq i \ne j \leq m)$\;
    $p_k:=\sum\limits_{}c_{ij}^kb_{ij}, 1 \leq k \leq n$\;
    $M:=\MAT{$c_{ij}^k$}$\;
    $\tilde{M}=\RREF{M}$\;
    $G:=\tilde{M} B$\;
    \uIf{\IB{$G$}}{
      $R:=Binomial$\;
    }
    \Else{$R:=NotBinomial$\;
    }
    \Return{$R$}\;
  }
\end{algorithm}

\paragraph*{Generalisation to Non-Reversible Networks}
The unconditional binomiality test via the binomial coefficient matrix
for a reversible chemical reaction network can be used as a subroutine
for testing unconditional binomiality of an arbitrary chemical
reaction network. In order to do so, partition a given chemical
reaction network $\mathcal{C}$ into a reversible reaction network
$\mathcal{C}_1$ and a non--reversible reaction network
$\mathcal{C}_2$.  Apply Algorithm \ref{alg:binom-test} to
$\mathcal{C}_1$, construct its binomial coefficient matrix, say $M_1$.
Construct the stoichiometric coefficient matrix of $\mathcal{C}_2$,
say $M_2$, and consider the block matrix $M:=(\tilde{M}_1|M_2)$.
Compute the row reduced echelon form of $M$, say $\tilde{M}$. If all
the rows of $\tilde{M}$ have at most one non-zero entry, then the
steady state ideal is binomial.

Otherwise, one can consider computing $\tilde{M}$ as a preprocessing
step and run another method, e.g., Gr\"obner bases, quantifier
elimination as in \cite{grigoriev2019efficiently}, or the method in
Dickenstein, et al \cite{millan2012chemical}.
  
\section{Complexity \& Comparisons}\label{sec:complexity}

\begin{proposition}
  Let $r$ be the number of reactions and $n$ be the number of species
  of a reversible chemical reaction network $\mathcal{C}$. The
  asymptotic worst case time complexity of testing unconditional
  binomiality of the steady state ideal of $\mathcal{C}$ via Algorithm
  \ref{alg:binom-test} can be bounded by
  $\mathcal{O}(\max(r,n)^\omega)$ where $\omega\approx 2.3737$, which
  is also the complexity of matrix multiplication.
\end{proposition}
\begin{proof}
  The operations in steps 1--4 and 7--11 are at most linear in terms
  of $r$ and $n$. Since $M$ is a matrix of size $n \times r$, where
  $r=|b_{ij}|$, and $B$ is a vector of size $r$, computing reduced row
  echelon form in step $5$ and also the matrix multiplication in step
  $6$ will cost at most $\mathcal{O}(\max(r,n)^\omega)$. Therefore the
  total number of operations in the algorithm can be bounded by
  $\mathcal{O}(\max(r,n)^\omega)$.
\end{proof}

In \cite[Section 4]{grigoriev2019efficiently} it has been shown that
there exists an exponential asymptotic worst case upper bound on the
time complexity of testing toricity. An immediate consequence of that
result is that the time complexity of testing binomiality can be
bounded by the same exponential function. Following the arguments in
\cite[Section 4]{grigoriev2019efficiently}, one can show that there
exists an algorithm for testing binomiality over
$\QQ[k_{ij},x_1,\dots,x_n]$ and $\QQ[x_1,\dots,x_n]$ simultaneously,
with an exponential upper bound for the worst case time complexity.
  
As mentioned earlier in Section \ref{sec:linearalgebra-algo}, the
reduced Gr\"obner basis of a binomial ideal, with respect to every
term order, includes only binomials. This directly can be seen from
running Buchberger's algorithm and that S--polynomials and their
reductions by binomials are binomial.  Eisenbud et. al' article
\cite{eisenbud1996binomial}, with a slightly different definition of
binomial ideals, investigates many properties of binomial ideals using
the latter fact.  Following this fact, a typical method for testing
binomiality is via computing a reduced Gr\"obner basis of a steady
state ideal $I \subseteq \QQ[k_{ij},x_1,\dots,x_n]$ The drawback of
computing Gr\"obner bases is that this is EXPSPACE-complete
\cite{mayr_complexity_1982}. So our algorithm is asymptotically
considerably more efficient than Gr\"obner basis computation.

\begin{example}[Models from the BioModels
  Repository\footnote{\url{https://www.ebi.ac.uk/biomodels/}}]
  \begin{itemize}
  \item There are twenty non--reversible biomodels in which Gr\"obner
    basis computations done in \cite{grigoriev2019efficiently} for
    testing conditional binomiality do not terminate in a six--hour
    time limit, however our algorithm terminates in less than three
    seconds.  Also there are six cases in which Gr\"obner basis
    computations terminate in less than six hours, but are at least
    1000 times slower than our algorithm. Finally there are ten models
    in which Gr\"obner basis is at least 500 times slower than our
    computations.
  \item There are sixty nine biomodels that are not considered for
    computation in \cite{grigoriev2019efficiently} because of the of
    the unclear numeric value of their rate constants. Our
    computations on almost all of those cases terminated in less than
    a second.
  \item (Reversible models from the BioModels Repository) Biomodels
    491 and 492 are both reversible. Biomodel 491 has 52 species and
    86 reactions.  The binomial coefficient matrix of this biomodel
    has size $52\times 86$ and has $\pm 1$ entries. A reduced row
    echelon form computation in Maple reveals in $0.344$ seconds that
    it is unconditionally binomial, while a Gr\"obner basis
    computation takes more than $12$ seconds to check its conditional
    binomiality. BioModel 492 has also 52 species, and includes 88
    reactions. The binomial coefficient matrix has entries $\pm1$ and
    is of size $52 \times 88$. This biomodel is also unconditionally
    binomial. It takes $0.25$ seconds for Maple to check its
    unconditional binomiality via Algorithm \ref{alg:binom-test} in
    Maple, while a Gr\"obner basis computation takes near $18$
    seconds, as one can see in the computations in \cite[Table
    3]{grigoriev2019efficiently}, which show the group structure of
    the steady state varieties of the models.
  \end{itemize}
\end{example}

Dickenstein et al.~in \cite{millan2012chemical} have proposed a method
for testing toricity of a chemical reaction network. The definitions
and purpose of that work are slightly different from our article,
hence comparisons between those two methods should be treated with
caution. While we focus on unconditional binomiality of the steady
state ideals of reversible reaction networks, , i.e., binomiality in
$\QQ[k_{ij},x_1,\dots,x_n]$, with the aim of efficiency of the
computations, the authors of the above article are interested in
conditional binomiality with algebraic dependencies between $k_{ij}$
such that the elimination ideal is binomial. Having mentioned that,
our method leads to the computation of reduced row echelon form of a
matrix of size $n\times r$ with integer entries which is polynomial
time, while Theorem 3.3. in \cite{millan2012chemical} requires
constructing a matrix of size $n \times s$ with entries from
$\ZZ[k_{ij}]$ and finding a particular partition of its kernel.

Considering Example 2.3 in \cite{perez_millan_chemical_2012}, our
algorithm constructs the matrix $M$ and its reduced row echelon form
$\tilde{M}$:
\begin{equation} 
  M=\begin{pmatrix}
    1 & 1 & -2 \cr
    -1 & -1 & 2 \cr
  \end{pmatrix},\quad
  \tilde{M}=\begin{pmatrix} 1 & 1 & 2 \cr 0 & 0 & 0 \cr
  \end{pmatrix},
\end{equation}
and we see that the steady state ideal is not unconditionally binomial
over $\QQ[k_{ij},x_1,\dots,x_n]$. The method in
\cite{perez_millan_chemical_2012} constructs
\begin{equation}
  \begin{pmatrix}
    -2k_{12}-k_{13} & 2k_{21}+k_{23} & k_{31}-k_{32} \cr
    2k_{12}+k_{13} & -2k_{21}-k_{23} & -k_{31}+k_{32} \cr
  \end{pmatrix},
\end{equation}
and finds an appropriate partition, which shows that the steady state
ideal is binomial in $\QQ[x_1,\dots,x_n]$ if and only if
$k_{31}=k_{32}$. As a larger example, consider the chemical reaction
network given in Example 3.13 in \cite{perez_millan_chemical_2012} and
assume that it is a reversible chemical reaction network. Our method
constructs a matrix with entries $\pm1$ of size $9\times 8$ and
computes its reduced row echelon form (in this case reduced row
echelon form, as entries are $\pm1$). The method described in
\cite{perez_millan_chemical_2012} leads to a $9\times 10$ matrix with
entries as linear polynomials in $\ZZ[k_{ij}]$ and computes a
particular partition of the kernel of the matrix.
  
For homogeneous ideals, Conradi and Kahle have shown in
\cite{conradi2015detecting} that the sufficient condition for
conditional binomiality in \cite{perez_millan_chemical_2012} is
necessary, too. Their Algorithm 3.3 tests conditional binomiality of a
homogeneous ideal, which can be generalised by homogenising. The
algorithm computes a basis for the ideal degree by degree and performs
reductions with respect to the computed basis elements at each degree
step.  Since our algorithm is intended for steady state ideals of
reversible chemical reaction networks, which are not necessarily
homogeneous, our following comparison with the Conradi--Kahle
algorithm bears a risk of being biased by homogenisation. We discuss
the execution of both algorithms on Example 3.15 in
\cite{perez_millan_chemical_2012}. This chemical reaction network does
not satisfy the sufficient condition presented in \cite[Theorem
3.3]{perez_millan_chemical_2012}. Testing this condition leads to the
construction of a $9 \times 13$ matrix with entries in $\ZZ[k_{ij}]$,
followed by further computations, including finding a particular
partition of its kernel.  Theorem 3.19 in
\cite{perez_millan_chemical_2012} is a generalisation of Theorem 3.3
there, which can test conditional binomiality of this example by
adding further rows and columns to the matrix. Conradi and Kahle also
treat this example with their algorithm. This requires the
construction of a coefficient matrix of size $9 \times 13$ with
entries in $\ZZ[k_{ij}]$ and certain row reductions. If we add
reactions so that the reaction network becomes reversible, our
algorithm will construct a matrix of size $9 \times 9$ with entries
$\pm 1$ and compute its reduced row echelon form to test unconditional
binomiality in $\QQ[k_{ij},x_1,\dots,x_9]$.

\section{Conclusions}\label{sec:conclusion}
Binomiality of steady state ideals is an interesting problem in
chemical reaction network theory. It has a rich history and literature
and is still an active research area. For instance, recently MESSI
systems have been introduced \cite{millan_structure_2018} following
the authors' work on binomiality of a system. Finding binomiality and
toricity is computationally hard from both a theoretical and a
practical point of view. It typically involves computations of
Gr\"obner bases, which is EXPSPACE-complete.

In a recent work \cite{grigoriev2019efficiently} we investigated
toricity of steady state varieties and gave efficient algorithms. In
particular, we experimentally investigated toricity of biological
models systematically via quantifier elimination. Besides that, we
presented exponential theoretical bounds on the toricity problem. The
current article, restricting to reversible reaction networks, aims at
an efficient linear algebra approach to the problem of unconditional
binomiality, which can be considered as a special case of the toricity
problem.

In that course, considering rate constants as indeterminates, we
assign a unique binomial to each reaction and construct the
coefficient matrix with respect to these binomials. Our algorithm
proposed here computes a reduced row echelon form of this matrix in
order to detect unconditional binomiality. The algorithm is quite
efficient, as it constructs comparatively small matrices whose entries
are integers. It is a polynomial time algorithm in terms of the number
of species and reactions. While other existing methods for testing
conditional binomiality have different settings and purposes than our
algorithm, for the common cases, our algorithm has advantages in terms
of efficiency.

\subsection*{Acknowledgments}
This work has been supported by the bilateral project ANR-17-CE40-0036
and DFG-391322026 SYMBIONT \cite{BoulierFages:18a,BoulierFages:18b}.

\end{document}